\documentclass[12pt]{article}

\usepackage{amsmath}
\usepackage[dvipsnames]{xcolor}
\usepackage{graphicx,psfrag,epsf}
\usepackage{enumerate}

\usepackage{tikz}

\usepackage{natbib}

\usepackage[lined,ruled,vlined]{algorithm2e}
\usepackage{url} 
\usepackage{mathtools,amssymb,amsthm}
\usepackage{soul} 
\usepackage{url} 
\usepackage{float}

\DeclareMathOperator\logit{logit}
\newcommand{\indep}{\perp \!\!\! \perp}

\newtheorem{property}{Property}

\newtheorem{theorem}{Theorem}
\newtheorem{corollary}{Corollary}
\newtheorem{lemma}{Lemma}

\title{Penalized regression with negative-unlabeled data: An approach to developing a long COVID research index}

\author{Harrison T. Reeder$^{1*}$, Tanayott Thaweethai$^{1}$ and Andrea S. Foulkes$^{1,2}$ \\ \\
$^{1}$Department of Medicine, Biostatistics \\ Massachusetts General Hospital and Harvard Medical School  \\ Boston, MA USA \\ \\ $^{2}$Department of Biostatistics \\ Harvard TH Chan School of Public Health \\ Boston, MA USA \\
*Contact email: hreeder@mgh.harvard.edu}

\date{}

\begin{document}
\def\spacingset#1{\renewcommand{\baselinestretch}%
{#1}\small\normalsize} \spacingset{1}

\maketitle

\abstract{Moderate to severe post-acute sequelae of SARS-CoV-2 infection (PASC), also called long COVID, is estimated to impact as many as 10\% of SARS-CoV-2 infected individuals, representing a chronic condition with a substantial global public health burden. An expansive literature has identified over 200 long-term and persistent symptoms associated with a history of SARS-CoV-2 infection; yet, there remains to be a clear consensus on a syndrome definition. Such a definition is a critical first step in future studies of risk and resiliency factors, mechanisms of disease, and interventions for both treatment and prevention. We recently applied a strategy for defining a PASC research index based on a Lasso-penalized logistic regression on history of SARS-CoV-2 infection. In the current paper we formalize and evaluate this approach through theoretical derivations and simulation studies. We demonstrate that this approach appropriately selects symptoms associated with PASC and results in a score that has high discriminatory power for detecting PASC. An application to data on participants enrolled in the RECOVER (Researching COVID to Enhance Recovery) Adult Cohort is presented to illustrate our findings. }

\newpage
\spacingset{1.45} 
\section{Introduction}
\label{sec:intro}

Globally, over 775 million cases of COVID-19 have been reported, resulting in over 7 million deaths \citep{covid19dashboard}. Among those who survive the acute phase of SARS-CoV-2 infection, many continue to experience long-term and often debilitating symptoms for months or even years, which collectively are known as long COVID, or post-acute sequelae of SARS-CoV-2 infection (PASC). Estimates of PASC incidence vary widely across studies, ranging from 5 to 25\% \citep{nittas2022long}, due to differences in study design as well as inconsistent definitions of PASC \citep{chou2024long}. 

In the absence of a gold-standard definition, we recently developed a symptom-based PASC research index in adults and children based on an informative subset of features \citep{pmid37278994, gross2024characterizing}. This approach uses a Lasso-penalized logistic regression model, with history of SARS-CoV-2 infection as the outcome to serve as a partial label for PASC status, and presence of long-term symptoms as features \citep{tibshirani1996regression}. The fitted model identifies features associated with PASC, and generates a quantitative index that tracks with the likelihood of PASC. The index derived using this approach has been used in studies of PASC in pregnant individuals \citep{metz2024post} and clinical laboratory abnormalities in PASC \citep{erlandson2024differentiation}. 

Using infection as a pseudo-label for PASC reflects the fact that individuals with a history of SARS-CoV-2 infection may or may not have PASC, while individuals without a history of SARS-CoV-2 infection cannot have PASC. As a result, infection history and PASC status form a partially-observed outcome referred to as ``negative-unlabeled'' data: those without history of infection are necessarily PASC ``negative'', while those with history of infection are ``unlabeled''. This is a variant of the more common ``positive-unlabeled'' data sometimes arising from electronic health records, in which some health characteristics are documented (``positive'') but lack of documentation does not mean the characteristic is absent (``unlabeled''). While existing methods for partially-labeled data commonly leverage subsets of individuals with definitive positive and negative documentation \citep{bekker2020learning}, in the absence of a gold-standard definition no individuals have a `positive' observed PASC status. However, because the population of infected individuals is a mixture of those with and without PASC, symptoms associated with PASC will also be more common among infected individuals, while symptoms not associated with PASC are expected to be indistinguishable between infected and uninfected individuals.

Nevertheless, it remains unclear whether an analysis strategy based on infection status as a pseudo-label can appropriately capture information about latent PASC status. Moreover, PASC symptoms exhibit correlation which may affect the selection performance of Lasso-penalized logistic regression, and therefore the resulting PASC index. Additionally, in practice studies of long COVID enroll infected and uninfected populations that may differ with respect to baseline covariates that are potentially associated with the presence of long-term symptoms---the resulting potential for confounding, and impact of techniques to mitigate confounding, remains unclear. 

To address these gaps, we investigate theoretical and operating characteristics of the PASC index using a rigorous conceptual framework for PASC onset and symptoms. We evaluate whether this approach correctly identifies a subset of symptoms that associate with PASC, and whether the PASC index has high discriminatory power for detecting PASC. We further explore balancing weights as a method for handling differential sampling by baseline demographics. Finally, we evaluate the relative performance of this approach compared to a simpler scoring framework defined by counting prevalent long-term symptoms. Evaluation of symptom count as an alternative to the PASC index reflects existing efforts to define PASC as inclusive of broad symptomatology after SARS-CoV-2 infection, such as recent emphasis by the National Academy of Sciences, Engineering, and Medicine that long COVID can manifest as any of hundreds of symptoms affecting every organ of the body \citep{NAP27768}.

There remains a pressing need for a standardized research definition that rigorously identifies patients with long COVID while maintaining a high degree of specificity \citep{brode2024long}. As such, arriving at a data-driven and discriminatory PASC index is a significant first step to further research into the risk and resiliency factors as well as the biological mechanisms of PASC and recovery.  
\section{Methods} \label{sec:methods}

\subsection{Approach} \label{subsec:approach}

PASC can be seen as a latent condition following infection that in turn affects the onset or persistence of long-term symptoms. In this context, let $Y$ be binary indicator for the presence of PASC and $A$ be a binary indicator for history of infection. We know that individuals without a history of infection cannot have PASC, while those who do have a history of infection will have a non-zero probability of PASC. We formalize this as follows:
\begin{property}\label{as:infpasc}
    Individuals without SARS-CoV-2 infection cannot develop PASC, i.e., $\Pr(Y=1 \mid A=0) = 0$.
\end{property}
Among those with SARS-CoV-2 infection, the probability of developing PASC is denoted $\Pr(Y=1\mid A=1)=\pi, \;\; 0 < \pi < 1$. The overall population probability of infection is denoted $Pr(A=1) = \alpha$. Let ${\mathbf X} = (X_1, \hdots, X_K)$ denote a vector of indicators for the presence of each of $K$ features (e.g., symptoms) that may depend on the presence of PASC. Further suppose a data generating model with a multiplicative effect of PASC status on symptom prevalence, 
\begin{align} \label{eq:xkgiveny}
    p_{k,y} = \Pr(X_{k}=1 \mid Y=y) = \beta_{0k} \left(\beta_{1k}\right)^{y},
\end{align}

\noindent for $k=1,\dots, K$ and $y=0,1$. For the $k$th feature, $\beta_{0k}$ corresponds to the prevalence among those without PASC, and $\beta_{1k}$ is the risk ratio for having the feature between those with and without PASC. By definition $0 < \beta_{0k} < 1$, and therefore to ensure $p_{k,y} \in (0,1)$, this implies that $0 < \beta_{1k} < 1/\beta_{0k}$.

Features can be further partitioned into $L$ groups, such that features in the same group may be dependent while features in different groups are independent. Formally, let $\mathcal{C}_l$ denote the set of indices of features in group $l$, for $l=1,\hdots,L$. We specifically consider Clayton copulas with common dependence parameter $\rho>0$ to admit positive dependence within each group. Formally, for a group of size $d$ define the Clayton copula function $C(u_1,\dots,u_{d} \mid \rho) = (u_1^{-\rho} + \dots + u_{d}^{-\rho} - d + 1)^{-1/\rho}$. Notating the marginal distribution function for feature $k$ as $F_k(x) = \Pr(X_{k} \leq x)$, the joint distribution of $X_{1},\dots,X_{d}$ is defined by $F(x_{1},\dots,x_{d}) = C(F_{1}(x_{1}), \dots, F_{d}(x_{d})).$ This retains the marginal distributions of each feature and their associations with PASC status, while inducing dependence between features within a group that increases as $\rho$ increases.


In general, interest lies in the $\beta_{1k}$'s of Equation~\eqref{eq:xkgiveny}, as they capture information on the relationship between PASC status and the probability of having each specific feature.  However, as $Y$ is unobservable, we are unable to estimate $\beta_{1k}$ directly. We therefore propose fitting a logistic regression model of the form,
\begin{align} \label{eq:modelA}
    \logit [\Pr(A=1)] = \gamma_0 + \gamma_1 X_1 + \hdots + \gamma_K X_K,
\end{align}
which yields a log-odds ratio $\gamma_k$ for each feature reflecting its association with infection status as a pseudo-label of PASC status.

To induce sparsity in the selected features, we estimate this model using a Lasso penalty with regularization parameter $\lambda$. In the simulations and data application presented below, this penalized model is fit via the \texttt{glmnet} package in R, using 10-fold cross-validated misclassification error to select $\lambda$ using the ``one-standard error rule'' implemented in the package \citep{friedman2009glmnet}.

Based on this model, we define the PASC index for individual $i$, $i=1,\hdots,N$ with observed features $x_{ik}$, as  
\begin{align} \label{eq:score}
    \mathcal{S}_i = \sum_{k=1}^K \widehat{\gamma}_k  x_{ik},
\end{align}

\noindent where $\widehat{\gamma}_k$ is the estimated coefficient of $X_k$. In practice, \cite{pmid37278994} round each $\widehat{\gamma}_k$ to the nearest tenth and then multiply by 10, yielding a PASC index that takes on integer values with a minimum of 0.

This approach reflects the hypotheses that the features selected in fitting the model of Equation~\ref{eq:modelA} will be associated with PASC, and that the PASC index of Equation~\ref{eq:score} will have strong discriminatory power for differentiating individuals with and without PASC. In the following section we motivate this strategy theoretically in the case of a single feature. We then show through simulation studies the performance of this index in discriminating true PASC status in the context of multiple features, $\mathbf{X}$.  

\subsection{Theoretical Results}

Assuming a single feature $X$, Equation~\eqref{eq:xkgiveny} reduces to
\begin{align} \label{eq:XgivenY}
    \Pr(X=1 \mid Y=y) = & \beta_0 \beta_1^{y}.
\end{align}
Here $\beta_1$ is the risk ratio of X for $Y=1$ versus $Y=0$. Again, in the absence of a gold-standard definition of PASC, $Y$ is not observable and we therefore can not directly estimate $\beta_1$. However, we do know that individuals without a history of infection can not have PASC, while those who do have a history of infection will have a non-zero probability of PASC. Reflecting the regression approach outlined by Equation~\eqref{eq:modelA}, we therefore consider the approximation of $\beta_1$ using the observable relationship between history of infection and symptom status, via the odds ratio
\begin{align}\label{eq:or}
OR_{A,X} = \left[\theta_1 / (1-\theta_1)\right]\left[\theta_0 / (1-\theta_0)\right],
\end{align}

\noindent where
\begin{align} \label{eq:aGivenx}
\theta_x = \Pr(A=1 \mid X=x), \;\; x=0,1. 
\end{align}

To formalize the conceptual structure relating infection, PASC status, and long-term symptoms, in addition to Property~\ref{as:infpasc} we state another defining characteristic of PASC:

\begin{property}\label{as:latent}
    Feature $X$ depends on prior infection status $A$ only through PASC status $Y$, i.e., $X \indep A \mid Y$.
\end{property}

\noindent Property~\ref{as:infpasc} reflects the definition of PASC as a condition that follows SARS-CoV-2 infection, and Property~\ref{as:latent} reflects the definition of PASC as a condition that encompasses the manifestation of all long-term sequelae following SARS-CoV-2 infection. Under these defining characteristics, the relationship between the observable odds ratio $OR_{A,X}$ defined in Equation~\eqref{eq:or} and the true risk ratio $\beta_1$ in Equation~\eqref{eq:XgivenY} can be represented as described in Theorem~\ref{thm:theorem1}.

\begin{theorem}\label{thm:theorem1}
    Denote population probability of infection $Pr(A=1) = \alpha$, PASC prevalence among the infected $\Pr(Y=1\mid A=1)=\pi$, baseline prevalence of $X$ denoted $\Pr(X=1\mid Y=0) = \beta_0$, and symptom risk ratio by PASC status $\beta_1$. Then for $OR_{A,X}$ as defined in Equation~\eqref{eq:or}, under Properties~\ref{as:infpasc} and \ref{as:latent} 
    \begin{align} \label{eq:or2}
    OR_{A,X} = 1 - \frac{\pi(1 - \beta_1)}{\beta_0  \pi (1-\beta_1) + (1-\beta_0)}.
    \end{align}
\end{theorem}

\begin{proof}
Applying Bayes' Rule with the Law of Total Probability in terms of $Y$ in the denominator,

\begin{equation} \label{eq:bayes}
\theta_x = \frac{\Pr(X=x \mid A=1) \Pr(A=1)}{\Pr(X=x \mid Y=1)\Pr(Y=1) + \Pr(X=x \mid Y=0)\Pr(Y=0)}.
\end{equation}

\noindent By Property~\ref{as:infpasc}, $\Pr(Y=1)=\Pr(Y=1 \mid A=1)\Pr(A=1)$, and therefore, $\Pr(Y=0)=\left[\Pr(A=1)^{-1}-\Pr(Y=1 \mid A=1)\right] \Pr(A=1)$. By Property~\ref{as:latent}, $\Pr(X=x \mid Y=y) = \Pr(X=x \mid Y=y, A=a)$. In turn, $\Pr(X=x \mid A=1)$ can be expressed as $\Pr(X=x \mid Y=1) \Pr(Y=1  \mid A=1) + \Pr(X=x \mid Y=0) \Pr(Y=0 \mid A=1)$.

Substituting these results into \eqref{eq:bayes} directly yields
\begin{align}
\theta_{1} = & \frac{ \beta_0\beta_1 \pi + \beta_0 (1-\pi) }{ \beta_0\beta_1 \pi + \beta_0 (\alpha^{-1} - \pi) } = \frac{ \beta_1 + (1-\pi)/\pi  }{ \beta_1 + (1 - \alpha\pi)/(\alpha\pi)  }
\end{align}
\noindent and
\begin{align}
\theta_{0} = & \frac{ (1-\beta_0\beta_1) \pi + (1-\beta_0) (1-\pi) }{ (1-\beta_0\beta_1) \pi + (1-\beta_0) (\alpha^{-1} - \pi) }.
\end{align}

\noindent The resulting relationship in Equation~\eqref{eq:or2} follows algebraically.
\end{proof}

This closed form does not have a directly intuitive interpretation, but it has several important implications explored below that support the use of the model in Equation~\eqref{eq:aGivenx} in the absence of $Y$ to characterizing the association between features and PASC. Notably, $OR_{A,X}$ does not depend on the population rate of infection $\alpha$, which simplifies further examination of this relationship. 


The most important implication of Theorem~\ref{thm:theorem1} is the capacity of $OR_{A,X}$ to inform knowledge of $\beta_1$, even though $\beta_1$ is unobservable. This is captured by the following corollary: 
\begin{corollary}\label{cor:cor1}
    Under the conditions of Theorem~\ref{thm:theorem1},
    \begin{itemize}
        \item[(a)] $OR_{A,X}=1 \Leftrightarrow \beta_1=1$,
        \item[(b)] $OR_{A,X}<1 \Leftrightarrow \beta_1<1$, and
        \item[(c)] $OR_{A,X}>1 \Leftrightarrow \beta_1>1$.
    \end{itemize}
\end{corollary}
\begin{proof}
The proof is given in Appendix A of the Supplementary Materials.
\end{proof}
This result emphasizes that under the stated conditions, there will be no observed association between $A$ and $X$ if and only if there is no underlying association between $Y$ and $X$. This motivates the selection of features based on the model in Equation~\eqref{eq:modelA} as representing identification of symptoms truly associated with $Y$. Moreover, the direction of association between $X$ and $Y$ is preserved in the association between $A$ and $X$. In other words, the direction of association between a feature and infection status should be the same as between the feature and PASC status.

While Corollary~\ref{cor:cor1} confirms that the direction of association is preserved, it is also of interest to understand the relative magnitude of $OR_{A,X}$ and $\beta_1$. In analyses with a binary outcome and binary exposure, it is known that the odds ratio will always be farther from the null than the risk ratio, but can approximate the risk ratio if the outcome is rare. Comparison of the relative magnitudes of $OR_{A,X}$ and $\beta_1$ is further complicated because they correspond to different outcomes. However, it remains of interest in particular when the odds ratio, $OR_{A,X}$, will numerically exaggerate the underlying risk ratio, $\beta_1$.

The distances of $\beta_1$ and $OR_{A,X}$ from the null can be characterized in absolute terms by $|\beta_1 - 1|$ and $|OR_{A,X}-1|$, and depend on $\pi$ and $\beta_0$ as follows: 
\begin{corollary}\label{cor:mag}
    Define $\phi = \left[(1-\pi)/\pi \right] \left[(1-\beta_0)/\beta_0\right]$. Then by Theorem~\ref{thm:theorem1},
    \begin{itemize}
        \item[(a)] if $\beta_1=\phi$ then $\beta_1 = OR_{A,X}$,
        \item[(b)] if $\beta_1>\phi$ then $|\beta_1-1| < |OR_{A,X}-1|$, i.e., $\beta_1$ is closer to 1 than the $OR_{A,X}$ is to 1, and
        \item[(c)] if $\beta_1<\phi$ then $|\beta_1-1| > |OR_{A,X}-1|$, i.e., the $OR_{A,X}$ is closer to 1 than $\beta_1$ is to 1.
    \end{itemize}
\end{corollary}
\begin{proof}
    The proof is given in Appendix A of the Supplementary Materials.
\end{proof}

\noindent Corollary~\ref{cor:mag} defines a threshold $\phi$ that depends only on the prevalence of PASC among the infected ($\pi$) and the baseline prevalence of the symptom ($\beta_0$). When $\beta_1$ is below this threshold, then $OR_{A,X}$ will be closer to 1 in magnitude than $\beta_1$, a form of `bias towards the null'. 

Importantly, for some values of $\pi$ and $\beta_0$, the parameter $\beta_1$ will always be strictly less than $\phi$, and therefore $OR_{A,X}$ will be biased towards the null for all possible $\beta_1$: 
\begin{corollary}\label{cor:suff}
    If $\beta_0 < 1 - [\pi/(1-\pi)]$, then $|\beta_1-1| > |OR_{A,X}-1|$ for all possible values of $\beta_1$, i.e., the region $0<\beta_1<1/\beta_0$.
\end{corollary}
\begin{proof}
    If $\beta_0 < 1 - [\pi/(1-\pi)]$, then it can be shown that $\phi > 1/\beta_0$. By definition in Equation~\eqref{eq:XgivenY}, the maximum value of $\beta_1$ is $1/\beta_0$, implying that $\phi > \beta_1$ for all $\beta_1$. By Corollary~\ref{cor:mag}, $\phi > \beta_1$ implies that $|\beta_1-1| > |OR_{A,X}-1|$, which then under these conditions holds for all possible values of $\beta_1$.
\end{proof}

Figure~\ref{fig:pibeta0settings} illustrates the region of $\pi$ and $\beta_0$ values satisfying Corollary~\ref{cor:suff}. The figure also highlights a useful sufficient condition for Corollary~\ref{cor:suff}, namely $\pi < 1/3$ and $\beta_0 < 1/2$: in practice, it is unlikely that the true rate of PASC among infected individuals is greater than one third, and generally symptoms of interest have prevalence among those without PASC of less than one half. Therefore, in most reasonable settings this condition will be met, and the $OR_{A,X}$ will bias towards the null regardless of the true value of $\beta_1$.

\begin{figure}
\begin{center}
    \includegraphics[width=0.5\textwidth]{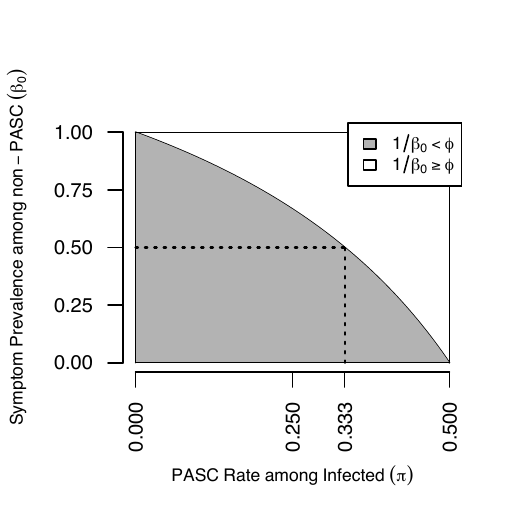}
\caption{Values of parameters satisfying Corollary~\ref{cor:suff}. Dotted lines characterize region meeting sufficient condition of $\pi<1/3$ and $\beta_0<0.5$.}
    \label{fig:pibeta0settings}
\end{center}    

\end{figure}

\subsection{Considerations to address differential sampling}
\label{sec:cov}

In this section we investigate the potential for differential sampling to affect the theoretical relationship derived in Equation~\eqref{eq:or2}, and assess the ability of a balancing weight approach \citep{pmid37278994} to reduce the impact of confounding. Notably, herein we are specifically investigating confounding of the relationship between infection and the features. We consider the potential for one or more additional factors, denoted $\mathbf{Z}$, to be associated with both the sampling of infected individuals, represented by $A$, and the presentation of features, given by $\mathbf{X}$. The additional complexities of considering potential confounding of the association between PASC status and features are discussed in Section~\ref{sec:disc}. 

To simplify presentation, suppose $\mathbf{Z}$ has discretely many strata. Using weights defined by
\begin{align}\label{eq:weight}
w(a,\mathbf{z}) = {\Pr(A=1\mid \mathbf{Z}=\mathbf{z})}/{\Pr(A=a\mid \mathbf{Z}=\mathbf{z})},
\end{align}
a weighted pseudopopulation is created with equal proportions of infected and uninfected individuals within each covariate stratum $\mathbf{z}$. That is, the distribution of potential confounding variables among  uninfected participants matches the distribution among infected participants. Combining these weights with the Law of Total Probability and the definition of conditional probability, the probability of infection given symptom status in the weighted pseudopopulation becomes
\begin{align}
    {\Pr}^{(w)}(A=a\mid \mathbf{X}=\mathbf{x}) = \frac{\sum_{\mathbf{z}} w(a,\mathbf{z})\cdot \Pr(\mathbf{X}=\mathbf{x},A=a,\mathbf{Z}=\mathbf{z})}{\sum_{a'=0}^1 \sum_{\mathbf{z}} w(a',\mathbf{z})\cdot \Pr(\mathbf{X}=\mathbf{x},A=a',\mathbf{Z}=\mathbf{z})}.
\end{align}
In this section we again investigate the theoretical properties using a single binary feature $X$. To formalize the role of $\mathbf{Z}$, we relax Property~\ref{as:latent} with the following:

\begin{property}
\label{as:latent2} For confounding variables $\mathbf{Z}$, infection status $A$, PASC status $Y$, and feature $X$, the following hold:
    \begin{itemize}
        \item[(a)] Feature $X$ depends on prior infection status $A$ only through latent PASC status $Y$ and possibly $\mathbf{Z}$, i.e., $X \indep A \mid Y,\mathbf{Z}$.
        \item[(b)] The probability of PASC given infection $\pi$ does not depend on $\mathbf{Z}$, i.e., $\pi = \Pr(Y=1\mid A=1) = \Pr(Y=1\mid A=1,\mathbf{Z}=\mathbf{z})$ for all $\mathbf{z}$.
    \end{itemize}
\end{property}

Property~\ref{as:latent2}a relaxes Property~\ref{as:latent}, by allowing additional dependence between $A$ and $X$ through an exogenous $\mathbf{Z}$. This may occur in practice if, for instance, a study samples infected and uninfected individuals differentially with respect to characteristics such as sex or age, and those characteristics are also independently associated with the presence of certain symptoms. Property~\ref{as:latent2}b encodes that PASC has a constant prevalence among people with SARS-CoV-2 regardless of $\mathbf{Z}$. 

Setting $\theta^{(w)}_x = {\Pr}^{(w)}(A=1\mid X=x)$, the resulting odds ratio estimated using balancing weights is given by,
\begin{align}
    OR^{(w)}_{A,X} = \frac{\theta^{(w)}_1 / \left(1-\theta^{(w)}_1\right)}{\theta^{(w)}_0 / \left(1-\theta^{(w)}_0\right)}.
\end{align}

\noindent Critically, it can be shown that this odds ratio is equivalent to a marginal adjusted odds ratio estimated using inverse probability of treatment weights, where infection $A$ is the treatment and feature $X$ is the outcome,
\begin{align}
    OR^{(w)}_{A,X} = \frac{\Pr(X=1\mid A=1)}{\Pr(X=0\mid A=1)} \bigg/ \frac{\sum_{\mathbf{z}} w(0,\mathbf{z})\cdot\Pr(X=1, A=0,\mathbf{Z}=\mathbf{z})}{\sum_{\mathbf{z}} w(0,\mathbf{z})\cdot\Pr(X=0, A=0,\mathbf{Z}=\mathbf{z})}.
\end{align}
In particular, this contrast would correspond under appropriate causal assumptions to a causal odds ratio reflecting the `average treatment effect among the treated'
\citep{karlson2023marginal,austin2011introduction}. This equivalence in the setting of a single symptom heuristically motivates the use of this balancing weight approach for the multivariable analysis in Equation~\eqref{eq:modelA}.

Unlike the closed form result of Theorem~\ref{thm:theorem1}, the correspondence between $\beta_1$ and $OR_{A,X}$ or $OR^{(w)}_{A,X}$ under Properties~\ref{as:infpasc} and \ref{as:latent2} does not have a simple closed form, as it depends not only on $\alpha$, $\pi$, and $\beta_0$, but also on how $\mathbf{Z}$ is distributed and associated with $X$ and $A$. However, the relationship can still be computed numerically for fixed values of these inputs.  For example, the relationship between $\beta_1$ and $OR_{A,X}$ and $OR^{(w)}_{A,X}$ is presented graphically in Figure~\ref{fig:ORRRconf} for a single binary confounder $Z$ having varying associations with $A$ and $X$, represented by the quantities $RR_{A,Z} = \Pr(A=1\mid Z=1)/\Pr(A=1\mid Z=0)$ and $RR_{X,Z} = \Pr(X=1\mid Z=1)/\Pr(X=1\mid Z=0)$. Other inputs were fixed at realistic values, as described in the figure text. For comparison, each graph is overlaid in grey with the relationship derived in Equation~\eqref{eq:or2} using the same inputs under no confounding, i.e., setting $RR_{A,Z} = RR_{X,Z}=1$. 

\begin{figure}
\begin{center}
	\includegraphics[width=\linewidth]{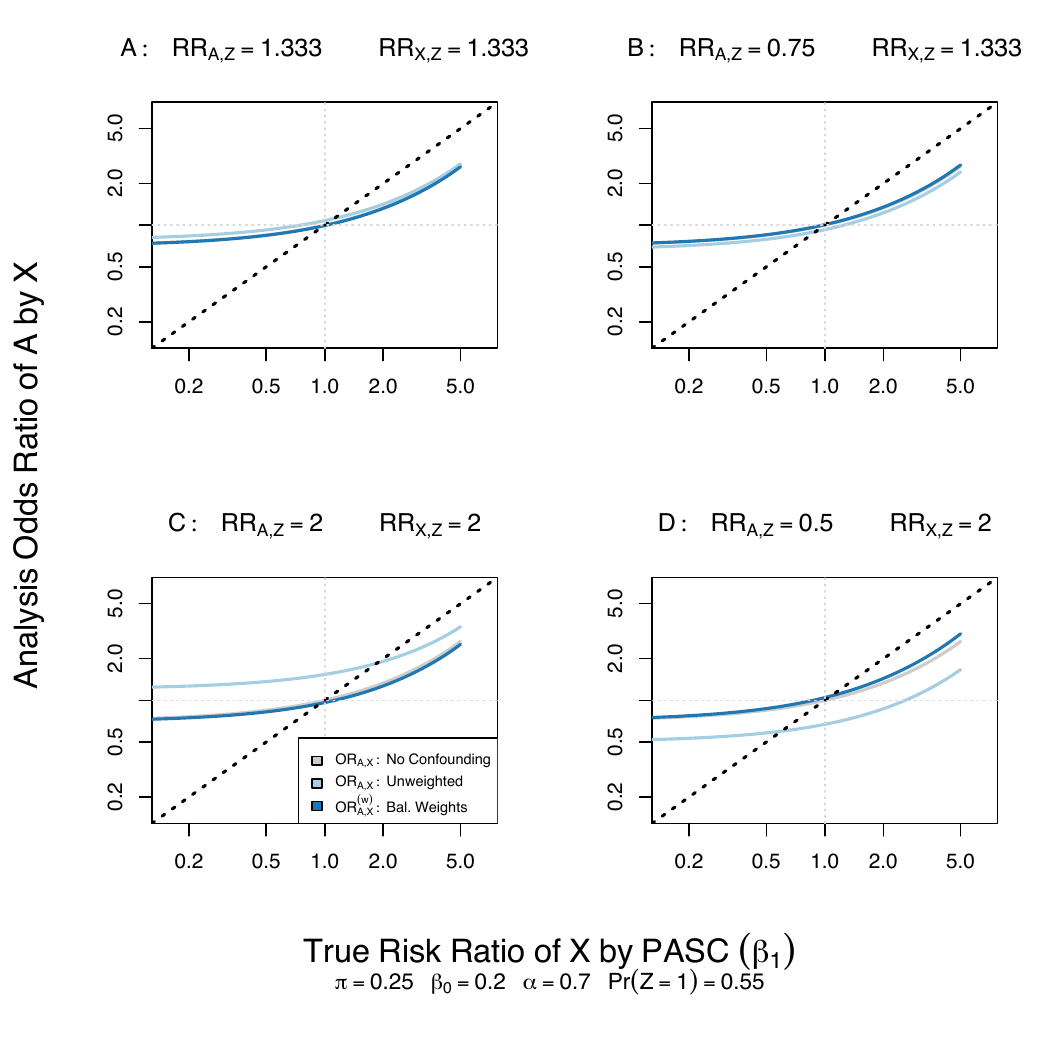}
	\caption{Relationship between quantities under Properties~\ref{as:infpasc} and \ref{as:latent2}. For comparison, $OR_{A,X}$ under `no confounding', i.e., Properties~\ref{as:infpasc} and \ref{as:latent}, also shown.}
 \label{fig:ORRRconf}
\end{center}
 
\end{figure}

Because the selected values of $\pi=0.25$ and $\beta_0=0.2$ satisfy Corollary~\ref{cor:suff}, the unconfounded theoretical relationship shows that $OR_{A,X}$ is always closer to 1 than $\beta_1$ is to 1. In the presence of confounding by $Z$, the unadjusted $OR_{A,X}$ becomes increasingly biased as $Z$ becomes more strongly associated with $X$ and $A$, as shown by the entire curve being shifted. However, across all confounder settings, the odds ratio estimated using balancing weights exhibits a very similar relationship to $\beta_1$ as would have occurred in the absence of confounding. This demonstrates the potential for balancing weighting to appropriately address bias due to covariate imbalances between infected and uninfected individuals.

\section{Simulation studies}
\label{sec:sim}

The preceding results focus on the special case of a single symptom to build intuition. However, in the presence of multiple symptoms the individually estimated log-odds ratios from the logistic regression model in Equation~\eqref{eq:score} will be conditional on other features, and no longer exactly correspond to $OR_{A,X}$ due to correlation between features and non-collapsibility of the odds ratio. Therefore we conducted simulation studies to further evaluate the discriminatory performance with respect to $Y$ of the PASC index given by Equation~\eqref{eq:score}, based on the full model of Equation~\eqref{eq:modelA}. Results in the absence of a confounding variable $\mathbf{Z}$ are presented here, with additional simulations examining the role of confounding presented in Appendix B of the supplementary materials. 

Data are generated according to the following parameter specifications: $\alpha=0.8$, $\pi=0.2$, $K=40$, and $\beta_{0k}= 0.2$ for $k=1,\dots,K$. A total of 12 symptoms were set to have true PASC associations: a `Low Signal' setting having $\beta_{1k} = 1.3$ for $k = 1,\dots,4$, $\beta_{1k} = 1.5$ for $k=5,\dots,8$, and $\beta_{1k}= 1.7$ for $k=9, \dots, 12$, and `Medium Signal' and `High Signal' settings exponentiating all effects by factors of 1.2 and 1.4 respectively. We set $\beta_{1k} = 1$ for $k=13,\dots,40$ representing features unrelated to PASC.

We also varied the correlation of symptoms, with an uncorrelated symptom setting ($\rho=0$), and then two settings with $\rho=5$ and symptoms in $L=9$ groups of varying sizes:
$|\mathcal{C}_1|= 1$, $|\mathcal{C}_2| = 2$, $|\mathcal{C}_3| = |\mathcal{C}_4| = |\mathcal{C}_5| = 3$, $|\mathcal{C}_6| = 5$, $|\mathcal{C}_7| = 6$, $|\mathcal{C}_8| = 7$, and $|\mathcal{C}_9|= 10$. One correlated symptom setting was `group sparse', meaning that for each symptom group $l$, either $\beta_{1k} = 1$ for all $k \in \mathcal{C}_l$, or $\beta_{1k} \neq 1$ for all $k \in \mathcal{C}_l$. The other correlated symptom setting was non group sparse, with groups having a combination of symptoms with zero and non-zero effects. 500 simulations were performed in each setting.

For each simulation we fit a Lasso-penalized logistic regression model as in \cite{pmid37278994} and described in Section~\ref{subsec:approach}. We report the selection performance of the Lasso-penalized model in identifying symptoms with a true PASC association, using the `true positive/negative rate' indicating the proportion of coefficient estimates that were correctly estimated as null / non-null. Higher values indicate improved selection. We also considered the Kendall's $\tau$ rank-correlation coefficient between the estimated coefficients $\widehat{\gamma}_k$ and the true values of $\beta_{1k}$, characterizing how closely the rank-order of coefficient estimates reflects to the true rank order of PASC association magnitudes.

For comparison, we considered an ad hoc score defined as the total number of symptoms an individual has. This comparator conceptually corresponds to approaches such as \citet{NAP27768} that include the presence of any symptoms when defining PASC. It mathematically corresponds to computation of the PASC index given by Equation~\eqref{eq:score} fixing $\widehat{\gamma}_k=1$ for all $k$. 

Finally, three metrics were used to evaluate performance of the proposed PASC index in discriminating between the true PASC status: 1) the area under the receiver operating characteristic curve (AUC); 2) the area under the precision-recall curve (AUCPR); and 3) the test statistic of a Wilcoxon rank sum test comparing the distributions of the PASC index between those truly with and without PASC.

Table~\ref{tab:sim1} reports the selection performance of the Lasso-penalized logistic regression model across simulation settings. Across low, medium, and high signal settings, the greatest number of symptoms were selected under the non-group sparse correlation setting, followed by the uncorrelated and then group sparse settings. This reflects that under non-group sparsity, symptoms with no true effect are highly correlated with symptoms that do have a true effect and therefore are more likely to be selected; this is evidenced by a decreased true negative rate. 

\begin{table}
\caption{Lasso feature selection performance}
\centering
\resizebox{0.9\textwidth}{!}{\begin{tabular}{lcccc}
  \hline
 &  &  &  & Est. vs. True Coef.  \\   
 & \# Features Selected & TPR & TNR & Rank-Correlation$^*$ \\ 
Setting  & Median (IQR) & Mean (SD) & Mean (SD) & Mean (SD)\\
  \hline
Low Signal &&& \\
~~~~Uncorrelated & 7 (3-14) & 0.486 (0.261) & 0.884 (0.163) & 0.530 (0.130) \\  
~~~~Non-Group Sparse & 14 (6-21) & 0.470 (0.235) & 0.697 (0.236)  & 0.137 (0.133) \\ 
~~~~Group Sparse & 4 (2-7) & 0.326 (0.197) & 0.916 (0.167) & 0.408 (0.120) \\
Medium Signal &&& \\
~~~~Uncorrelated & 8 (4-12) & 0.549 (0.250) & 0.905 (0.151) & 0.607 (0.122) \\ 
~~~~Non-Group Sparse & 15 (9-21) & 0.523 (0.218) & 0.676 (0.219) & 0.130 (0.128) \\ 
~~~~Group Sparse & 4 (2-6) & 0.325 (0.179) & 0.944 (0.133) & 0.447 (0.118) \\ 
High Signal &&& \\
~~~~Uncorrelated & 8 (5-11) & 0.598 (0.234) & 0.933 (0.119) & 0.672 (0.123) \\ 
~~~~Non-Group Sparse& 15 (8-22) & 0.545 (0.215) & 0.681 (0.221) & 0.135 (0.141) \\ 
~~~~Group Sparse & 4 (3-6) & 0.343 (0.175) & 0.958 (0.108) & 0.478 (0.115) \\ 
   \hline
\end{tabular}}

\footnotesize Abbreviations: TPR, True Positive Rate; TNR, True Negative Rate.

*Kendall's $\tau$ between estimated regression coefficients and true risk ratios of PASC vs each symptom.
\label{tab:sim1}
\end{table}

Importantly, the Lasso-penalized model tended to be conservative in variable selection, with generally high true negative rates in the uncorrelated and group sparse settings, but modest true positive rates that also varied depending on symptom correlation. We also observed that the estimated coefficients were relatively highly rank-correlated with the true PASC associations. However, under non-group sparse correlated symptoms the rank-correlation between estimated and true associations weakened, as symptoms with no true effect were more likely to have a nonzero coefficient estimate. 

\begin{figure}
\begin{center}
    \includegraphics[width=0.85\textwidth,page=1]{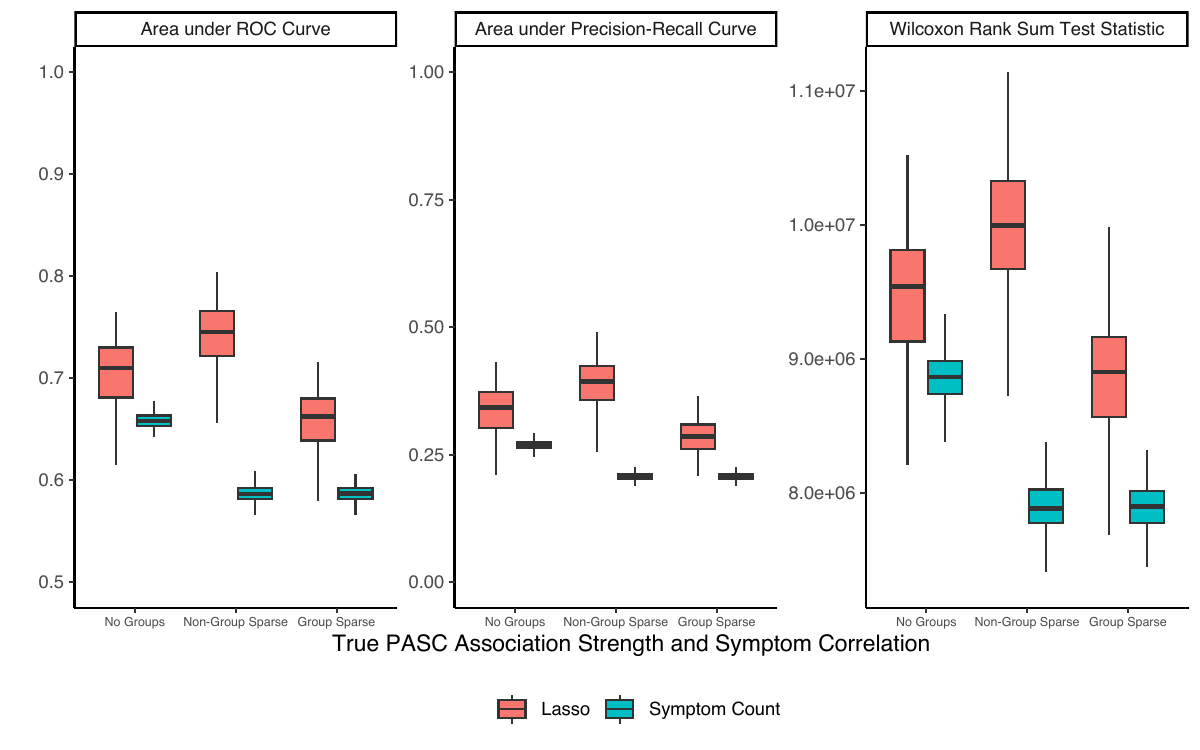}
    \caption{Discriminative performance of PASC index. Simulations are based on `Medium Signal' effect strength. } \label{fig:sim1}      
\end{center}
\end{figure}

As shown in Figure~\ref{fig:sim1}, the PASC index derived from Lasso-penalized logistic regression on infection status showed reasonable capacity to discriminate between those with and without PASC, across all forms of correlation. The figure focuses on the `Medium Signal' setting, with analogous results under varying signal strengths (results not shown).
Compared to the impact symptom correlation had on symptom selection, the impact of symptom correlation on PASC index discrimination was minimal for the regression-based PASC index. This seems to indicate that PASC index performance does not substantially degrade even when only some true non-zero symptoms are selected into the index, particularly under group sparsity in which PASC-associated symptoms have strong correlation.

Across all settings, the symptom count consistently underperformed relative to the Lasso-penalized logistic regression approach. In particular, the symptom count notably degraded in performance in the presence of correlation between symptoms.

\section{Application to RECOVER-Adult}
\label{sec:dataapp}

The NIH-sponsored Researching COVID to Enhance Recovery (RECOVER) Adult Cohort (RECOVER-Adult) is an ongoing observational cohort study of adults age 18 and older, with and without a history of SARS-CoV-2 \citep{pmid37352211}. Participants complete surveys recording the presence of symptoms at enrollment and at three-month follow-up intervals. Study time is relative to a time origin defined as the date of an individual's first SARS-CoV-2 infection, or date of a negative test result for uninfected participants. Participants enrolled at any time within 3 years of this date, and for the purpose of this example we considered presence of symptoms at the first study visit at least 6 months after infection/negative test for each participant. Recruitment began on December 1, 2021, with $n=9764$ contributing an eligible study visit at the time of data lock on April 10, 2023. A subset of $n=9702$ had complete data on sex assigned at birth, age, and race/ethnicity and were used in the analysis. 

A summary of cohort demographics by infection status is provided in Appendix C of the supplementary materials. As the demographic differences between infected and uninfected participants may also be associated with symptoms, we used the balancing weighted approach proposed in Section~\ref{sec:cov}. Weights were defined as in Equation~\eqref{eq:weight}, with a total of 30 strata defined by age (18-45, 46-65, $>$65), sex (male, female), and race (non-Hispanic White, non-Hispanic Black, non-Hispanic Asian, Hispanic, Mixed race/other/missing).

A total of 37 candidate symptoms from 12 clinically-defined categories were considered in the cross-sectional analysis, including cardiac, gastrointestinal, neurologic, musculoskeletal, and others. The full symptom list, with symptom rates by infection status, is available in Appendix C of the supplementary materials. The lefthand column of Figure~\ref{fig:fits} shows the symptoms selected using the Lasso-penalized logistic regression with balancing weights. A total of 13 symptoms were selected, with the symptoms with largest coefficients---and therefore largest contributions to the PASC index---being loss or change of smell/taste, post-exertional malaise (PEM), chronic cough, and brain fog. These symptoms have indeed become hallmark symptoms of PASC \citep{pmid37278994}.

There were small differences in the estimates with and without balancing weights (Figure~\ref{fig:fits}). The set of selected symptoms was the same with the exception of shortness of breath (selected by the unadjusted model only) and abnormal movements (selected by the balancing-weighted model only). The differences in coefficients were generally small, with the largest absolute difference for brain fog (balancing weighted: 0.330 vs unadjusted: 0.497). While both models selected the same top 4 symptoms the next largest coefficients in the unadjusted model were heart palpitations and dizziness, 
compared to thirst and chest pain in the balancing weighted model.

\begin{figure}
\begin{center}
    \includegraphics[width=\textwidth]{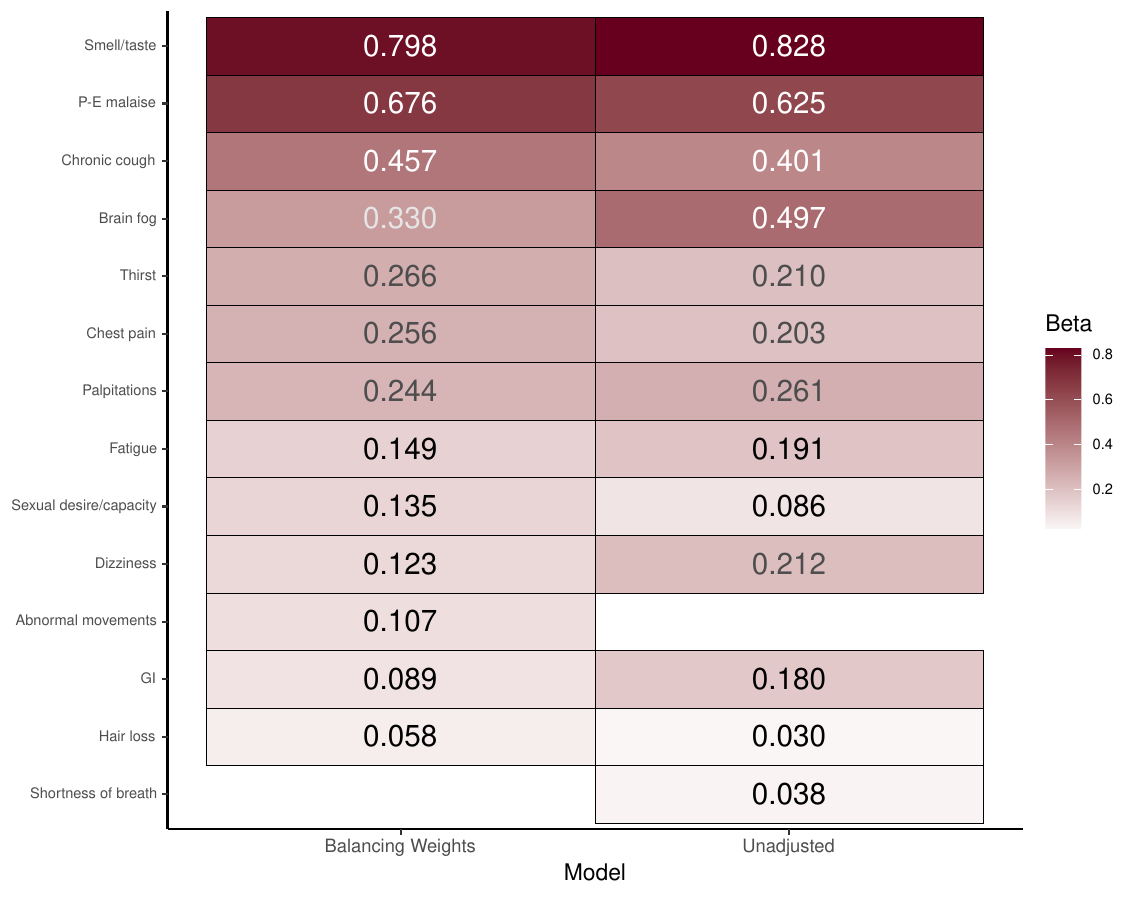}
    \caption{Lasso-penalized logistic regression coefficient estimates using RECOVER-Adult data.}
    \label{fig:fits}    
\end{center}
\end{figure}

Using the coefficients estimated under the balancing weighted analysis, we calculated a PASC index for each individual as in Section~\ref{subsec:approach}, resulting in an index ranging from 0-37 (Appendix C, Figure C.1). Among those with history of infection, 2933 (34.1\%) had an index of 0, compared to 711 (64.1\%) among those without history of infection. The median and interquartile range (IQR) of the index were 2 (0-11) and 1 (0-1), respectively. The correspondence between the PASC index and count of symptoms is illustrated in Appendix C, Figure C.2. Among those with history of infection, 2173 (25.3\%) had a count of 0 symptoms, compared to 521 (47.0\%) among those without history of infection. The median (IQR) symptom counts were 3 (0-8) and 1 (0-3), respectively. 

In practice, the PASC index and symptom count can be thresholded to classify individuals as having PASC. To illustrate the result of thresholding in RECOVER-Adult, we report the rate of PASC among the uninfected for fixed rates within infected participants (Figure~\ref{fig:roc}). Estimates of this curve are presented using 10-fold cross validation, where we fit the Lasso-penalized logistic regression model in each possible subset of 90\% of the data and estimate the index and resulting curve in the remaining 10\%. At any fixed rate of PASC classification among the infected, the corresponding rate of PASC classification in the uninfected will be lower using the PASC index than using symptom count. For example, for classification of PASC among the infected at a fixed rate of 23\% as reported previously, the estimated rate of PASC classification of uninfected individuals in an independent sample is 3.8\% (95\% CI: 2.9-4.8\%), compared to 6.4\% (95\% CI: 4.6-8.2\%) using symptom count. 

\begin{figure}
\begin{center}
    \includegraphics[width=0.8\textwidth]{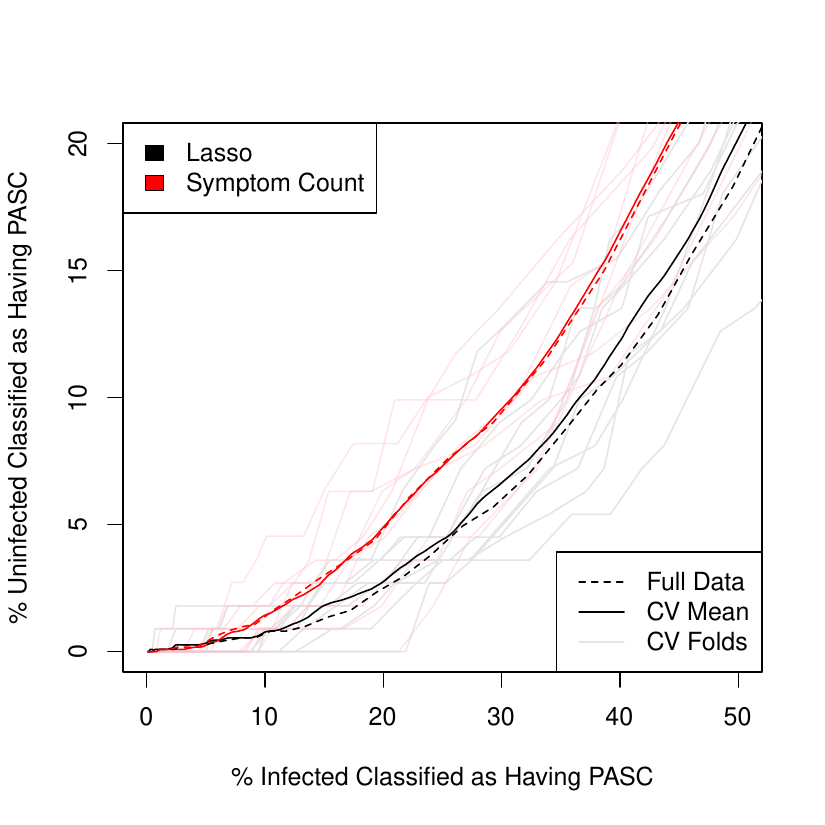}
    \caption{Characteristics of classification rules. Results shown within the same full data sample (dashed), and using 10-fold cross-validation (CV), averaging the folds (solid) and for each individual fold (light).}
    \label{fig:roc}    
\end{center}
\end{figure}
\section{Discussion}
\label{sec:disc}

This manuscript provides strong evidence in support of using Lasso-penalized logistic regression, using history of infection as a pseudo label, to identify PASC-associated features and define an associated PASC index. In the absence of a gold standard definition, an important marker of discriminative performance for PASC indices is the ability to identify PASC among infected participants, with minimal misclassification of uninfected participants as having PASC. We demonstrate that the index based on Lasso-penalized logistic regression, compared to a count of symptoms, leads to fewer misclassifications of uninfected participants as having PASC while capturing the same proportion of PASC cases among infected participants. We also found consistent evidence, through simulation results and theoretical results under a single feature (Figure~\ref{fig:ORRRconf}), that differences between weighted and unweighted analyses require both a relatively large imbalance in the stratifying variables by infection status, and a relatively strong association between the stratifying variables and symptoms. Finally, we demonstrated by simulation that the performance of the regression-based PASC index is robust to complex correlation structures between symptoms.

Several alternative but related strategies to the one described herein are viable. For example, in the example provided, model features were defined as the presence of symptoms, while more generally, these features could capture information on incident symptoms, worsening of symptoms, or other health-related characteristics. A machine learning type approach, such as random forests \citep{breiman2001random} or support vector machines \citep{boser1992}, could also be used in place of Lasso-penalized regression, offering the additional advantage of allowing for complex interactions and conditional associations. However, using one of these alternative supervised learning strategies does not inherently address the core challenge posed by the `negative-unlabeled' data inherent in studying PASC absent a gold-standard definition. Further extensions might also consider differential rates of PASC by levels of a confounding variable as well as interactions between symptoms and demographic characteristics on rates of PASC. Finally, while this manuscript demonstrated the utility of a PASC index, further consideration of how this index could potentially be used to develop a PASC classification is an important next step. 


\section*{Funding}

Support for this project was provided by NIH/NHLBI R01HL162373.


\bibliographystyle{apalike}
\bibliography{sourcesshorter}

\newpage

\appendix 

\makeatletter
\renewcommand{\thetable}{\thesection.\@arabic\c@table}
\@addtoreset{table}{section}
\makeatother

\makeatletter
\renewcommand{\thefigure}{\thesection.\@arabic\c@figure}
\@addtoreset{figure}{section}
\makeatother

\makeatletter
\renewcommand{\theequation}{\thesection.\@arabic\c@equation}
\@addtoreset{equation}{section}
\makeatother

\makeatletter
\renewcommand{\thelemma}{\thesection.\@arabic\c@lemma}
\@addtoreset{lemma}{section}
\makeatother


\section*{Appendix Introduction}
In these supplementary materials we present additional details and results beyond what could be presented in the main manuscript. To distinguish the two sections, alpha-numeric labels are used in the supplementary materials while numeric labels are used in the main paper. 
Appendix~\ref{app:proofcorr} provides proofs of theoretical results. 
Appendix~\ref{app:confsim} provides additional simulation results. 
Appendix~\ref{app:data} provides additional results from the data application. 

\section{Proofs of Theoretical Results}\label{app:proofcorr}
Because in the present context $\beta_1$ is a risk ratio multiplying a baseline risk $0<\beta_0<1$, then $0 < \beta_0\beta_1 \leq 1$ and equivalently $\beta_1 \leq 1/\beta_0$. We begin with a short result implied by this constraint, that is used in the proofs of Corollaries 1 and 2. 
\begin{lemma}\label{lem:temp1}
    For constants $0 < \pi < 1$ and $0 < \beta_0 < 1$, if $0 < \beta_1 \leq 1/\beta_0$, then 
    \begin{equation}
        \beta_0\pi (1 - \beta_1) + (1-\beta_0) > 0.
    \end{equation}
\end{lemma}
\begin{proof} Rearranging the initial constraint $\beta_1 < 1/\beta_1$, the result follows:
    \begin{align}
        1 - \beta_1 &\geq 1 - 1/\beta_0,
        \\ \Rightarrow \beta_0\pi (1 - \beta_1) + (1-\beta_0) & \geq \pi(\beta_0 - 1) + (1 - \beta_0),
        \\ \Rightarrow \beta_0\pi (1 - \beta_1) + (1-\beta_0) & \geq (1-\pi)(1-\beta_0) > 0.
    \end{align}

\end{proof}

\subsection{Proof of Corollary 1}
\begin{proof}[Proof of Corollary 1]
    The proof of part (a) follows trivially by assuming either $OR_{A,X}=1$ or $\beta_1=1$.

    For part (b), we first prove the reverse direction by assuming $\beta_1 < 1$. Because $0 < \pi < 1$, it directly follows that $0 < \pi(1-\beta_1) < 1$. Then using the additional fact that $0 < \beta_0 < 1$, this implies the following:
   \begin{align}
        \beta_0\pi(1 - \beta_1) & < \beta_0 ,
        \\ \Rightarrow  \beta_0\pi(1 - \beta_1) + (1 - \beta_0) & < 1,
        \\ \Rightarrow  1- \frac{\pi(1-\beta_1)}{\beta_0\pi(1 - \beta_1) + (1 - \beta_0)} & < 1- \pi(1-\beta_1) < 1,
    \end{align}
    and therefore $OR_{A,X} < 1$.

    Correspondingly, to prove the reverse direction of part (c), we set $\beta_1 > 1$. As before, it directly follows that $1-\beta_1 < 0$, and therefore that $\pi(1-\beta_1) < 0$. This further implies that
    \begin{align}
        \beta_0\pi (1-\beta_1) + (1-\beta_0) & < 1-\beta_0,
        \\ \Rightarrow \frac{\pi(1-\beta_1)}{\beta_0\pi (1-\beta_1) + (1-\beta_0)} & < \frac{\pi(1-\beta_1)}{1-\beta_0},
        \\ \Rightarrow 1 - \frac{\pi(1-\beta_1)}{\beta_0\pi (1-\beta_1) + (1-\beta_0)} &> 1 - \frac{\pi(1-\beta_1)}{1-\beta_0} > 1,
    \end{align}
    and therefore $OR_{A,X} > 1$.

    To prove the forward direction of part (b) we assume $OR_{A,X}<1$, rearranged as 
    \begin{align}
        \frac{\pi(1 - \beta_1)}{\beta_0 (1-\beta_1) \pi + (1-\beta_0)} &> 0. \label{eq:forwardb}
    \end{align}
    By Lemma~\ref{lem:temp1} the denominator of the lefthand side of \eqref{eq:forwardb} must be positive, so therefore the numerator must also be positive. This yields the final result that $\beta_1 < 1$.


    Finally, to prove the forward direction of part (c), we assume $OR_{A,X}>1$, which yields
    \begin{align}
        \frac{\pi(1 - \beta_1)}{\beta_0 (1-\beta_1) \pi + (1-\beta_0)} & < 0. \label{eq:forwardc}
    \end{align}
    Again by Lemma~\ref{lem:temp1} the denominator of the lefthand side of \eqref{eq:forwardc} must be positive, so therefore the numerator must be negative, implying the result that $\beta_1 > 1$.

\end{proof}

\subsection{Proof of Corollary 2}
We begin with another short lemma confirming Corollary 2(a), and laying out steps used in the more general proof below.
\begin{lemma}\label{lem:equal}
    Under the conditions of Theorem 1, $\beta_1=OR_{A,X}$ if and only if $\beta_1=\phi$.
\end{lemma}
\begin{proof}
If $\beta_1=1$, then the result follows directly from Corollary 1. Therefore, we focus on the case of $\beta_1 \neq 1$. Then, $\beta_1 = \phi$ can be written
       \begin{align}
        \beta_1 & = 1 - \frac{\pi - (1-\beta_0)}{\beta_0\pi}, \label{eq:start}
        \\ \Rightarrow \pi & = \beta_0 (1-\beta_1) \pi + (1-\beta_0),
        \\ \Rightarrow \frac{\pi}{\beta_0 (1-\beta_1) \pi + (1-\beta_0)} & = 1, \label{eq:step1}
        \\ \Rightarrow \frac{\pi(1 - \beta_1)}{\beta_0 (1-\beta_1) \pi + (1-\beta_0)} & = 1-\beta_1,  \label{eq:step2}
        \\ \Rightarrow \beta_1 & =  1 - \frac{\pi(1 - \beta_1)}{\beta_0 (1-\beta_1) \pi + (1-\beta_0)} = OR_{A,X}. \label{eq:final}
    \end{align} 

\end{proof}

\begin{proof}[Proof of Corollary 2]

Proof of part (a) is shown by Lemma~\ref{lem:equal}. Proof of part (b) begins by setting $\beta_1 > \phi$ and deriving the corresponding relationship between $\beta_1$ and $OR_{A,X}$. Proof of part (c) proceeds analogously setting $\beta_1 < \phi$. This can be done by changing Equation~\eqref{eq:start} into an inequality, and determining the direction of the resulting inequality in Equation~\eqref{eq:final}. 

The inequality at step~\eqref{eq:step1} will not flip after dividing by $[\beta_0 (1-\beta_1) \pi + (1-\beta_0)]$ because by Lemma~\ref{lem:temp1} this expression must be positive. However, at step~\eqref{eq:step2} the inequality may flip depending on the sign of $(1-\beta_1)$. Therefore, we proceed as follows to establish the result case by case:
\begin{itemize}
    \item[(i)] Fix $\phi$ relative to 1, and then fix $\beta_1$ relative to $\phi$ and 1. 
    \item[(ii)] Establish the resulting relationship between $\beta_1$ and $OR_{A,X}$, and thus the relative magnitude of $|\beta_1 - 1|$ and $|OR_{A,X}-1|$
\end{itemize}
Table~\ref{tab:my_label} presents the casewise results, which taken together establish Corollary 2.

\begin{table}[H]
    \centering
    \caption{Casewise proof results of Corollary 2.}
    \label{tab:my_label}    
    \resizebox{1.0\textwidth}{!}{\begin{tabular}{lcccc}\hline
&  &  & Relationship of  & Relationship of \\
Case & Condition on $\phi$ & Condition on $\beta_1$ & $\beta_1$ and $OR_{A,X}$ & $|\beta_1 - 1|$ and $|OR_{A,X}-1|$ \\\hline
   1 &  $(0 < \phi < 1) \Leftrightarrow (1-\pi < \beta_0)$    &  $\beta_1< \phi < 1$   & $\beta_1 < OR_{A,X}$ & $|\beta_1 - 1| > |OR_{A,X}-1|$ \\
   2 &  $(0 < \phi < 1) \Leftrightarrow (1-\pi < \beta_0)$    &  $ \phi < \beta < 1$   & $\beta_1 > OR_{A,X}$ & $|\beta_1 - 1| < |OR_{A,X}-1|$ \\
   3 &  $(0 < \phi < 1) \Leftrightarrow (1-\pi < \beta_0)$    &  $ \phi < 1 < \beta$   & $\beta_1 < OR_{A,X}$ & $|\beta_1 - 1| < |OR_{A,X}-1|$ \\
   4 &  $(\phi = 1) \Leftrightarrow (1-\pi = \beta_0)$    &  $\beta_1 < \phi = 1$   & $\beta_1 < OR_{A,X}$ & $|\beta_1 - 1| > |OR_{A,X}-1|$ \\
   5 &  $(\phi = 1) \Leftrightarrow (1-\pi = \beta_0)$    &  $1 = \phi < \beta_1$   & $\beta_1 < OR_{A,X}$ & $|\beta_1 - 1| < |OR_{A,X}-1|$ \\
   6 &  $(1 < \phi \leq 1/\beta_0)$ &  $\beta_1 < 1 < \phi$   & $\beta_1 < OR_{A,X}$ & $|\beta_1 - 1| > |OR_{A,X}-1|$ \\
   7 & $(1 < \phi \leq 1/\beta_0)$ &  $1 < \beta_1 < \phi$   & $\beta_1 > OR_{A,X}$ & $|\beta_1 - 1| > |OR_{A,X}-1|$ \\
   8 & $(1 < \phi \leq 1/\beta_0)$ &  $1 < \phi < \beta_1$   & $\beta_1 < OR_{A,X}$ & $|\beta_1 - 1| < |OR_{A,X}-1|$ \\
   9 & $(1 < 1/\beta_0 < \phi )$ &  $\beta_1 < 1 < \phi$   & $\beta_1 < OR_{A,X}$ & $|\beta_1 - 1| > |OR_{A,X}-1|$ \\
   10 & $(1 < 1/\beta_0 < \phi )$ &  $1 < \beta_1 < \phi$   & $\beta_1 > OR_{A,X}$ & $|\beta_1 - 1| > |OR_{A,X}-1|$ \\\hline
    \end{tabular}}
\end{table}

\end{proof}

\newpage

\section{Additional Simulation Results}\label{app:confsim}

To more systematically evaluate the impact of balancing weighting on performance in the presence of confounding, simulation studies are conducted to evaluate discrimination with respect to $Y$ of the PASC score given by Equation (3).
Data were generated with the same parameters as the above simulations, except with $\alpha=0.7$ corresponding to a population with 70\% having had SARS-CoV-2 infection, to allow for differences in infection prevalence by confounder level. Marginally $P(Z=1) = 0.55$, and the marginal risk ratio for $A$ by $Z$ was allowed to be either 1.7, or 0.65, to capture negative or positive Z/A associations. Twelve symptoms had a non-zero association with $Z$ with marginal risk ratios varying between 1.35 and 2, and the set of symptoms associated with $Z$ was set to either the 12 symptoms with a true PASC association (`Z/X Overlapping'), or 12 other symptoms without true PASC association (`Z/X Non-Overlapping'). Finally, we considered a setting in which no symptoms were associated with $Z$ marginally. Finally, the PASC-symptom associations were fixed as in the 'Medium Signal' setting above, yielding a total of 6 simulation settings. The metrics and comparators used were the same as previously.

Table~\ref{tab:sim2} reports the selection performance of lasso-penalized logistic regression models with and without balancing weights, across simulation settings. In settings with no Z/X association, there is no mechanism for confounding, and the results are comparable, with the unadjusted approach actually performing slightly better due to increased efficiency. However, in the presence of both Z/X and Z/A associations, the model fit with balancing weights substantially outperforms the unadjusted model, due to the impact of confounding. When Z was associated with non-PASC symptoms, the unadjusted model included many such symptoms, leading to larger model size and lower true negative rate. When Z was associated with the same PASC-associated symptoms, but the associations were in opposing directions, the unadjusted model tended to miss truly non-zero symptoms due to confounding, leading to decreased true positive rate and poor rank-correlation between estimated and true coefficients.

Figure~\ref{fig:confsim1} compares the discrimination statistics between the models in these settings, revealing similar patterns. In the absence of a Z/A association, and therefore no impact of confounding, performance was comparable, with differences due to decreased efficiency of the weighted approach. When confounding affected a non-overlapping set of symptoms, the unadjusted model included more such symptoms associated with infection due to confounding but not truly associated with PASC, leading to degraded performance. The balancing weighted approach better recovered the true PASC-associated symptoms, and therefore the PASC score achieved better discrimination. Finally, when the true PASC-associated symptoms were also subject to confounding, the impact of weighting depended on the direction of confounding bias. In the artificial setting of `Z/X Overlapping, Positive Z/A Association' in which all PASC-associated symptoms also had confounding bias in the same direction, the result was that the unadjusted approach performed better, because the true signal was amplified by confounding. More realistically, when confounding bias counteracted the direction of PASC association as in `Z/X Overlapping, Negative Z/A Association', then the unadjusted approach was unable to select relevant symptoms, leading to degraded performance.

\begin{table}[H]
\caption{Lasso symptom coefficient selection performance}
\centering
\resizebox{1.0\textwidth}{!}{\begin{tabular}{lllllllll}
  \hline
&&&&&&&\multicolumn{2}{c}{Est. vs. True}\\
& \multicolumn{2}{c}{Selected} & \multicolumn{2}{c}{TPR} & \multicolumn{2}{c}{TNR} & \multicolumn{2}{c}{Coef. Rank-Correlation$^*$}  \\
& \multicolumn{2}{c}{Median (IQR)} & \multicolumn{2}{c}{Mean (SD)} & \multicolumn{2}{c}{Mean (SD)} & \multicolumn{2}{c}{Mean (SD)}  \\
& Bal. Weights & Unadj. & Bal. Weights & Unadj. & Bal. Weights & Unadj.  & Bal. Weights & Unadj. \\ 
  \hline
  No Z/X Assoc.
Neg. Z/A Assoc. & 8 (4-13) & 8 (5-12) & 0.541 (0.25) & 0.598 (0.244) & 0.894 (0.156) & 0.925 (0.135) & 0.588 (0.123) & 0.663 (0.132) \\ 
No Z/X Assoc.
Pos. Z/A Assoc. & 8 (3.75-13) & 8 (5-12) & 0.526 (0.259) & 0.598 (0.244) & 0.89 (0.17) & 0.925 (0.135) & 0.576 (0.125) & 0.663 (0.132) \\ 
  Z/X Non-Over.
Neg. Z/A Assoc. & 8 (4-12) & 19 (16-22) & 0.532 (0.25) & 0.637 (0.197) & 0.902 (0.153) & 0.597 (0.101) & 0.589 (0.126) & 0.614 (0.066) \\ 
  Z/X Non-Over.
Pos. Z/A Assoc. & 7 (3-12) & 19 (16-22) & 0.51 (0.255) & 0.613 (0.195) & 0.903 (0.148) & 0.585 (0.083) & 0.571 (0.128) & 0.126 (0.1) \\ 
  Z/X Overlapping
Neg. Z/A Assoc. & 6 (3-10) & 5 (3-8) & 0.469 (0.242) & 0.382 (0.194) & 0.938 (0.13) & 0.946 (0.109) & 0.581 (0.143) & -0.235 (0.113) \\ 
  Z/X Overlapping
Pos. Z/A Assoc. & 6 (3-11) & 11 (11-12) & 0.486 (0.261) & 0.929 (0.077) & 0.925 (0.151) & 0.995 (0.028) & 0.579 (0.15) & 0.838 (0.045) \\ 
   \hline
\end{tabular}}
\footnotesize Abbreviations: TNR, True Negative Rate; TPR, True Positive Rate.

*Kendall's $\tau$ between estimated regression coefficients and true Symptom-PASC risk ratios.
\label{tab:sim2}
\end{table}

\begin{figure}[H]
    \centering
    \includegraphics[width=0.9\textwidth,page=1]{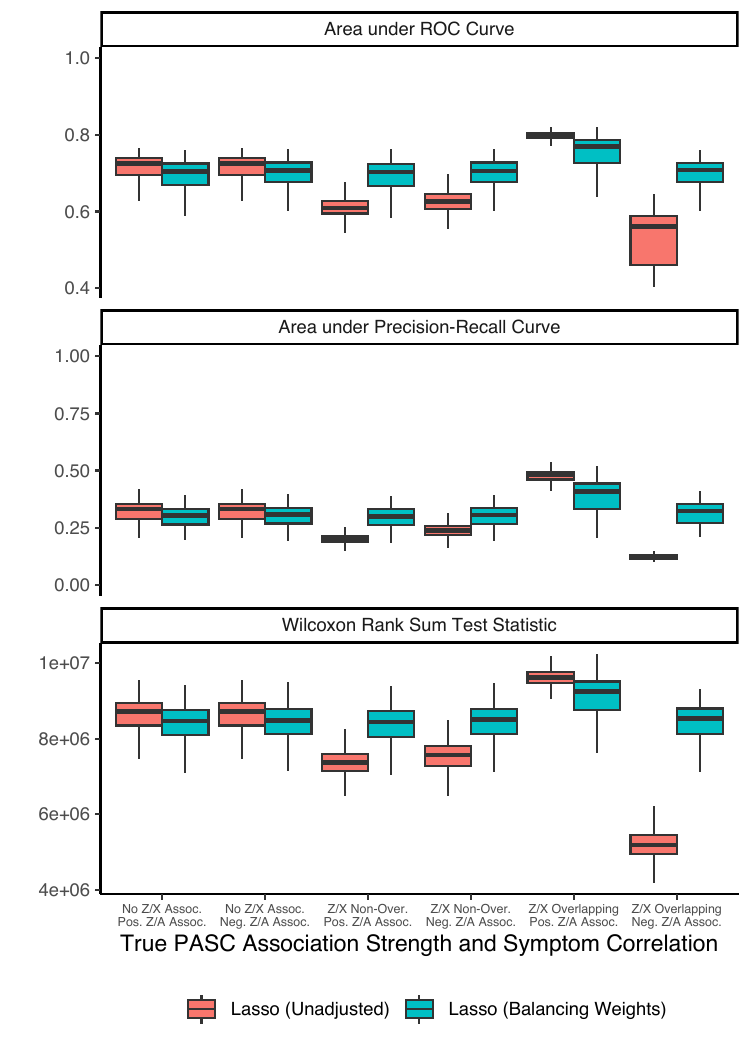}
    \caption{Discriminative performance of estimated PASC scores in distinguishing true PASC status under confounding.}
    \label{fig:confsim1}
\end{figure}

\section{Additional Data Application Results}\label{app:data}

\begin{table}[H]
\centering
\caption{Analysis cohort demographic characteristics by infection status.}
\begin{tabular}{llll}
  \hline
& Overall & Infected & Uninfected \\ 
Category & ($n=9702$) & ($n=8593$) & ($n=1109$) \\ 
  \hline
\textit{Sex assigned at birth} \\
  ~~Male & 2773 (28.6\%) & 2374 (27.6\%) & 399 (36\%) \\ 
  ~~Female & 6929 (71.4\%) & 6219 (72.4\%) & 710 (64\%) \\ 
\textit{Age category} \\
  ~~18-45 & 4745 (48.9\%) & 4369 (50.8\%) & 376 (33.9\%) \\ 
  ~~46-65 & 3656 (37.7\%) & 3158 (36.8\%) & 498 (44.9\%) \\ 
  ~~$>$65 & 1301 (13.4\%) & 1066 (12.4\%) & 235 (21.2\%) \\ 
\textit{Race/Ethnicity} \\
  ~~Non-Hispanic White & 5705 (58.8\%) & 5021 (58.4\%) & 684 (61.7\%) \\ 
  ~~Non-Hispanic Black & 1416 (14.6\%) & 1219 (14.2\%) & 197 (17.8\%) \\ 
  ~~Non-Hispanic Asian & 501 (5.2\%) & 428 (5\%) & 73 (6.6\%) \\ 
  ~~Hispanic & 1591 (16.4\%) & 1472 (17.1\%) & 119 (10.7\%) \\ 
  ~~Mixed race/Other/Missing & 489 (5\%) & 453 (5.3\%) & 36 (3.2\%) \\ 
  \hline
\end{tabular}
\end{table}

\begin{table}[H]
\centering
\caption{Symptom relative frequencies by infection status.}
\resizebox{0.85\textwidth}{!}{\begin{tabular}{lllll}
  \hline
 & & & & Uninfected \\ 
Group & Symptom & Infected, \% & Uninfected, \% & (Balancing Weights), \% \\ 
  \hline
General & Fatigue & 38.4 & 16.7 & 17.9 \\ 
   & Fever/sweats/chills & 12 & 3.7 & 3.9 \\ 
   & P-E malaise & 27.7 & 6.9 & 7.2 \\ 
   & Swelling of legs & 11 & 5.9 & 5.1 \\ 
  Cardiac & Chest pain & 7.8 & 1.1 & 1 \\ 
   & Palpitations & 21 & 6.7 & 7.1 \\ 
  Dermatologic & Hair loss & 18.9 & 9.9 & 10 \\ 
   & Skin color changes & 7 & 2.4 & 2.9 \\ 
   & Skin pain & 2.8 & 1 & 1.1 \\ 
   & Skin rash & 8.4 & 4.8 & 5 \\ 
  Eye & Vision & 10 & 3.2 & 3 \\ 
  Ear & Hearing & 21.3 & 13.7 & 11.1 \\ 
  Gastrointestinal & Abdominal pain & 5.2 & 2 & 2.8 \\ 
   & Dry mouth & 14.5 & 6 & 5.1 \\ 
   & GI & 25.3 & 10 & 11.1 \\ 
   & Teeth & 12.1 & 6.9 & 5.8 \\ 
  Metabolic & Thirst & 13.9 & 4 & 4 \\ 
  Musculoskeletal & Back pain & 15.2 & 9.2 & 8.9 \\ 
   & Foot pain & 7.6 & 4.4 & 4.3 \\ 
   & Joint pain & 16.6 & 9.6 & 8.5 \\ 
   & Muscle pain & 14 & 5.8 & 5.9 \\ 
   & Weakness & 13.4 & 5 & 4.5 \\ 
  Neurologic & Abnormal movements & 4.4 & 0.7 & 0.6 \\ 
   & Brain fog & 20.4 & 4.4 & 5.4 \\ 
   & Dizziness & 22.9 & 7.1 & 8.1 \\ 
   & Headache & 13.1 & 3.4 & 3.7 \\ 
   & Smell/taste & 12.6 & 1.9 & 2 \\ 
   & Tremor & 6.7 & 3 & 2.6 \\ 
  Psychiatric & Anxiety & 11.5 & 6 & 6.9 \\ 
   & Depression & 11.3 & 5.1 & 4.9 \\ 
   & Sexual desire/capacity & 18.8 & 8.1 & 8.1 \\ 
   & Sleep disturbance & 11.8 & 3.8 & 3.7 \\ 
  Respiratory & Chronic cough & 11.7 & 3.2 & 3 \\ 
   & Shortness of breath & 11.2 & 2.3 & 2.5 \\ 
   & Sleep apnea & 17.5 & 10 & 9.5 \\ 
   & Throat pain & 3 & 0.5 & 0.6 \\ 
  Urinary & Bladder & 13 & 7.3 & 7.1 \\ 
   \hline
\end{tabular}}
\end{table}

\begin{figure}[H]
    \centering
    \includegraphics[width=1\textwidth]{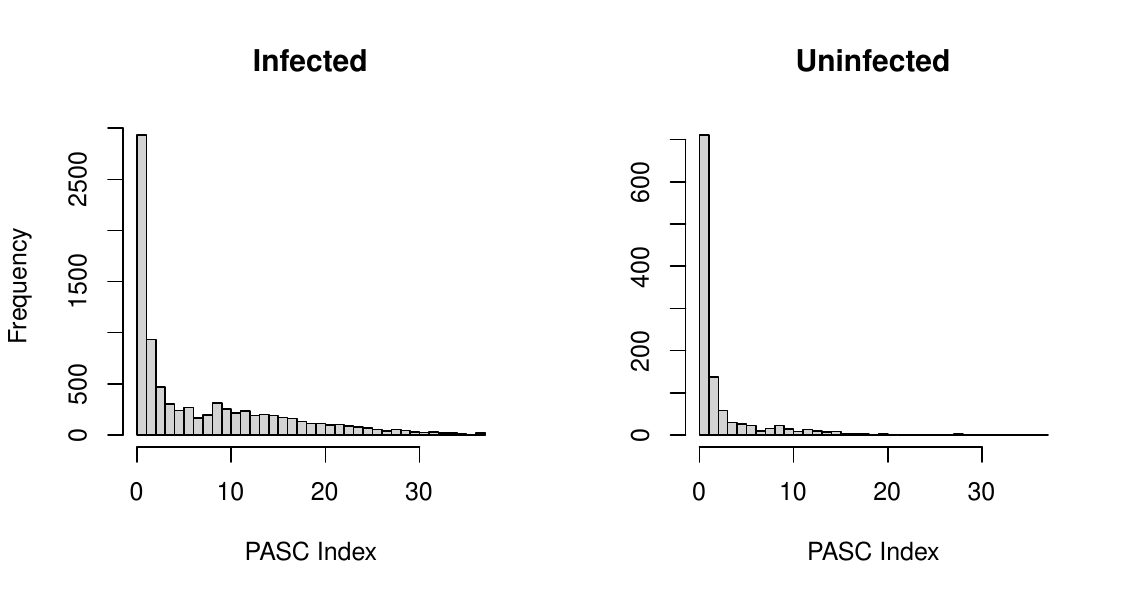}
    \caption{Histogram of PASC index derived from lasso-penalized logistic regression with balancing weights, stratified by infection status.}
    \label{fig:hist}
\end{figure}

\begin{figure}[H]
    \centering
    \includegraphics[width=0.9\textwidth]{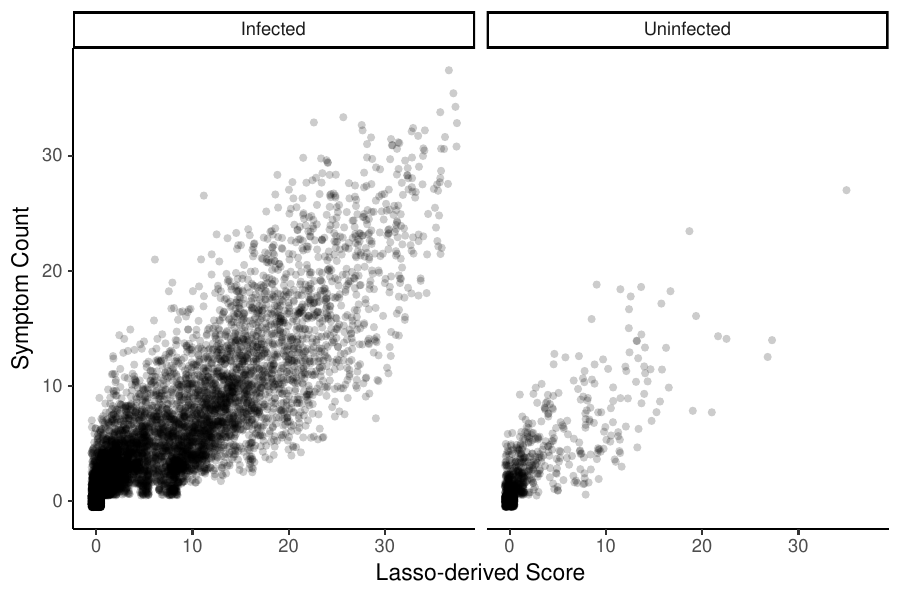}
    \caption{Jittered scatterplot of PASC index derived from lasso-penalized logistic regression with balancing weights versus symptom count, stratified by infection status.}
    \label{fig:scat}
\end{figure}

\end{document}